\documentclass{article}
\usepackage[utf8]{inputenc}
\usepackage[inner=3cm,outer=3cm,bottom=2cm]{geometry}

\usepackage{amsmath, amsthm, amssymb, mathtools}
\usepackage{lmodern}
\usepackage{stmaryrd}
\usepackage{xfrac}
\usepackage{dsfont}
\usepackage[hidelinks]{hyperref}
\usepackage{cleveref}
\usepackage{tcolorbox}
\usepackage{float}
\usepackage{thmtools}
\usepackage{thm-restate}

\title{Towards Stronger Counterexamples to the Log-Approximate-Rank Conjecture}

\author{Arkadev Chattopadhyay \footnote{Tata Institute for Fundamental Research, Mumbai. email: arkadev.c@tifr.res.in} \and %
Ankit Garg \footnote{Microsoft Research India, Bengaluru. email: garga@microsoft.com} \and %
Suhail Sherif \footnote{Tata Institute for Fundamental Research, Mumbai. This work was mainly done while the author was at Microsoft Research India, Bengaluru. email: suhail.sherif@gmail.com}
}

\date{}

\newtheorem{theorem}{Theorem}[section]
\newtheorem{lemma}[theorem]{Lemma}

\newtheorem{corollary}[theorem]{Corollary}
\newtheorem{claim}[theorem]{Claim}
\newtheorem{remark}[theorem]{Remark}
\newtheorem{fact}[theorem]{Fact}
\newtheorem{conjecture}[theorem]{Conjecture}
\newtheorem{definition}[theorem]{Definition}

\newcommand{\zone}{\{0,1\}}
\newcommand{\pmone}{\{-1,1\}}
\newcommand{\mV}{\mathcal{V}}

\newcommand{\mU}{\mathcal{U}}
\newcommand{\R}{\mathbb{R}}

\newcommand{\E}{\mathbb{E}}
\newcommand{\F}{\mathbb{F}}
\newcommand{\half}{\sfrac{1}{2}}

\newcommand{\bra}[1]{\{#1\}}
\newcommand{\spa}{\mathrm{span}}
\newcommand{\dimension}{\mathsf{dim}}
\newcommand{\codimension}{\mathsf{codim}}
\newcommand{\cosetmap}{\mathsf{coset}}

\newcommand{\defeq}{\mathrel{\mathop:}=}

\newcommand{\XOR}{\mathsf{XOR}}

\newcommand{\sparsity}[1]{\|\hat{#1}\|_{0}}
\newcommand{\asparsity}[2]{\|\hat{#1}\|_{0,#2}}

\newcommand{\spectral}[1]{\left\|\widehat{#1}\right\|_{1}}
\newcommand{\aspectral}[2]{\left\|\widehat{#1}\right\|_{1,#2}}

\newcommand{\RPDT}[1]{\mathsf{R^{\oplus}_{#1}}}

\newcommand{\dotp}[2]{\langle #1, #2 \rangle}

\newcommand{\SINK}{\mathsf{SINK}}

\newcommand{\A}{h}
\begin{document}

\maketitle

\begin{abstract}
    We give improved separations for the query complexity analogue of the log-approximate-rank conjecture i.e. we show that there are a plethora of total Boolean functions on $n$ input bits, each of which has approximate Fourier sparsity at most $O(n^3)$ and randomized parity decision tree complexity $\Theta(n)$. This improves upon the recent work of Chattopadhyay, Mande and Sherif \cite{CMS20} both qualitatively (in terms of designing a large number of examples) and quantitatively (improving the gap from quartic to cubic). We leave open the problem of proving a randomized communication complexity lower bound for XOR compositions of our examples. A linear lower bound would lead to new and improved refutations of the log-approximate-rank conjecture. Moreover, if any of these compositions had even a sub-linear cost randomized communication protocol, it would demonstrate that randomized parity decision tree complexity does not lift to randomized communication complexity in general (with the XOR gadget). 
\end{abstract}

\section{Introduction}
The Log-Rank Conjecture (LRC) of Lovasz and Saks asserts that two very seemingly different quantities, one the deterministic communication complexity of a total function $f$ (denoted by $D(f)$) and the other the log of the  rank of its communication matrix (denoted by $(M_f)$) over the field of reals, are essentially the same, i.e. within a fixed polynomial of each other. While this thirty year old conjecture remains wide open, it's natural to try upper-bounding the communication complexity of $f$ by \emph{some} function of the rank of $M_f$. The best such known bound was obtained by Lovett~\cite{Lov16}, rather recently, which showed that $D(f)$ is at most the square-root of the rank of $M_f$, ignoring log factors.

A tempting analog of the LRC for randomized communication complexity appears in a book by Lee and Shraibman~\cite{LS09} where it was named as the Log-Approximate-Rank Conjecture (LARC). Informally, this is LRC with deterministic communication complexity replaced by bounded-error randomized complexity of $f$, and rank replaced by the \emph{approximate rank} of $M_f$, where the approximation is uniform point-wise. The LARC is important for several reasons. First, it implies the LRC itself~\cite{GL14}. Second, it implies several other central conjectures, like the polynomial equivalence of quantum and classical communication complexity of total functions ~\cite{BdW01}. Third, every known lower bound, until very recently, was no larger than a small polynomial of the log of the approximate rank. Very recently, Chattopadhyay, Mande and Sherif \cite{CMS20} provided a surprisingly simple counterexample to the LARC, that exponentially separated randomized communication complexity from the log of the approximate rank. In particular, their function $f$ has Alice and Bob holding $n$ bits each, the approximate rank of its $2^n \times 2^n$ communication matrix  $M_f$ is merely $O(n^2)$ and yet the randomized communication complexity is $\Theta(\sqrt{n})$. 

Some questions immediately arise from the above refutation of the LARC. First, is the refutation optimal? There are two ways to measure optimality. The approximate rank and communication complexity are separated by a 4th power. Is this separation true for all functions i.e. is randomized communication complexity always upper bounded by fourth-root of the approximate rank? Interestingly, G{\'a}l and Syed \cite{GS19} recently showed that quantum communication complexity is upper bounded by at most square-root of the approximate rank but for randomized communication, the best upper bound is still linear in the approximate rank. The second way to view optimality is the extent of the gap achieved between log of the approximate rank and communication complexity. This is $O(\log n)$ vs. $\sqrt{n}$ for the current refutation. Can this gap be widened via other functions? This leads us to, of course, the related problem of finding other counter-examples to LARC. 
Finding a richer set of counter-examples, besides being interesting in their own right, could prove useful for understanding other central conjectures. A concrete example is the question of relative power of quantum and classical protocols to solve total functions, a major open problem. If we have to find a total function with an exponential gap between the quantum communication and randomized communication complexities (if one exists at all), then the function should also have an exponential separation between log of approximate rank and randomized communication complexity.\footnote{Since log of the approximate rank lower bounds quantum communication as well.} However, it was shown by Anshu et al. \cite{ABT19} and Sinha and de Wolf \cite{SdW19} that the function of \cite{CMS20} has large quantum communication complexity (hence refuting the quantum version of LARC as well). This motivates the search for other examples refuting the LARC.


In this work, we come up with a rich set of functions that leaves us with the following win-win situation: either every one of these functions gives a stronger refutation of the LARC than what is known or there is no \emph{lifting theorem} for randomized communication complexity of XOR functions. Lifting theorems, in the setting of communication complexity, lift the complexity of a function $f$ in an appropriate query model to the communication complexity of a problem crafted out of $f$ naturally by block composition with a gadget $g$, denoted by $f \circ g$. Starting with the celebrated work of Raz and McKenzie \cite{RM97}, they have enabled major progress recently in communication complexity and adjoining areas \cite{GPW18,dRMNPRV19,GJW18,CKLM18}. In all these theorems, the size of the gadget $g$ is at least logarithmic in the input length of the query function $f$. A challenging open problem is to prove lifting theorems for a constant size gadget. A natural one is the one bit\footnote{the gadget size here means the number of bits held by each of the two players.} XOR gadget. It is not hard to verify that a (randomized) parity decision tree (R)PDT algorithm for $f$ of cost $c$ readily translates into a communication protocol of cost $2c$ for $f \circ \text{XOR}$. A lifting theorem for XOR functions would assert the converse. In other words, a communication protocol cannot be more efficient than naively simulating the optimal RPDT. The strongest evidence for such an assertion is the result of Hatami, Hosseini and Lovett \cite{HHL16} who show that if $f$ has deterministic PDT cost $c$, then $f \circ \text{XOR}$ has deterministic communication complexity $c^{\Omega(1)}$. While no general result exists for the randomized model, the community believes it to be plausible. We state our main result informally.

\begin{theorem}[Informal]
Assuming XOR lifting theorems for randomized communication complexity, there exists a rich class of functions $f:\{0,1\}^n \to \{0,1\}$, such that $M_{f \circ \text{XOR}}$ has approximate rank $O(n^3)$ and $R(f\circ \text{XOR}) = \Theta(n)$.
\end{theorem}

Thus, conditionally, we get the following improvements over the results in \cite{CMS20}: (1) We narrow the gap between approximate rank and randomized communication complexity from quartic to cubic.  (2) We expand the gap between log-approximate-rank and randomized complexity from $O(\log n)$ vs. $\sqrt{n}$ to $O(\log n)$ vs. $n$, thus yielding essentially the strongest possible refutation of the LARC, under plausible assumptions. While this is a nice conceptual way to view our results, it seems proving communication lower bounds for these functions will require new tools and techniques. On the other hand, coming up with non-trivial communication protocol for any of these functions will rule out a PDT to communication lifting theorem for XOR functions in the randomized model.


\subsection{Main Ideas}
The starting point of our work is to pursue the idea in \cite{CMS20} of looking for functions with small (approximate) spectral norm, i.e. functions whose sum of the magnitude of Fourier coefficients is a small polynomial in $n$. 
The previous counterexample to the LARC used the concept of disjoint subcubes to achieve this as every subcube has spectral norm one.  This implied that a function $f$ whose set of ones form a union of polynomially many disjoint subcubes will have polynomial spectral norm. The fact that polynomial spectral norm implies polynomial Fourier sparsity, yields that the approximate rank of every such $f$ lifted by XOR is guaranteed to be small. The randomized communication complexity of one such function, $\SINK \circ \XOR$, was shown to be large via a Corruption Bound, the proof of which utilized Shearer's Lemma. The randomized parity decision tree lower bound used a robust subspace-hitting property of the subcubes instead.

In this work, we study a broader class of functions based on disjoint subspaces. The approximate rank of their lifts by XOR is again guaranteed to be small. The main conceptual contribution of our work is to identify a property that is sufficient for every such union of subspaces to have large RPDT complexity. 
Remarkably, this property is quite well encapsulated in the concept of Subspace Designs, a notion that has been studied in the literature in the context of error correcting codes and pseudorandomness \cite{GX12,GK16, GXY17}. We show that Subspace Designs are hard for RPDTs. The general philosophy of LARC like conjectures is that randomized complexity of total functions is well captured/characterized by algebraic or analytical measures of the function like (approximate) rank. For instance, a classical result of Nisan and Szegedy \cite{NS92} confirms this idea in the world of randomized (and quantum) query complexity where the relevant algebraic measure is approximate degree. In the world of PDTs, the natural algebraic notion is approximate Fourier sparsity. The work of \cite{CMS20} refuted this philosophy for parity decision trees via the SINK function, whose approximate Fourier sparsity is $O(n^2)$ and RPDT complexity is $\Theta(\sqrt{n})$. Our lower bounds for functions based on subspace designs yields unconditionally a stronger refutation of this philosophy for the model of parity decision trees. We state here our result in terms of random subspaces because this yields the cleanest formulation.

\begin{restatable}[Main Result]{theorem}{separation}
\label{thm:separation}
 Let $m = 100n$. Let $\mV = \bra{V_1, V_2, \dots, V_m}$ be a set of subspaces of $\zone^n$ chosen independently and uniformly at random from the set of subspaces of dimension $2n/5$. Let $f$ be the function that outputs $1$ on the set $\bigcup_{V \in \mV} V$. With probability $1-o(1)$ the following two statements are true.
    \begin{itemize}
        \item Randomized parity decision tree complexity of $f$ is at least $\Omega(n)$.
        \item The spectral norm of $f$ (sum of absolute values of its Fourier coefficients) is upper bounded by $O(n)$ and its approximate Fourier sparsity is upper bounded by $O(n^3)$.
    \end{itemize}
    Hence there exist functions which have a merely cubic gap between approximate Fourier sparsity and RPDT complexity.
\end{restatable}


The two properties of random subspaces appearing in such a collection that we use are the following: each pair of them have no non-trivial intersection. They also form a (dual) subspace design.
We are not able to prove non-trivial lower bounds for the communication problems arising out of Subspace Designs composed with the XOR gadget. However, in Section~\ref{subsec:communication}, we state concrete conjectures, that seem to be interesting from a Fourier analytic and additive combinatorics point of view, which imply linear lower bounds for such communication problems.

\subsection{Organization and plan of the paper}

Section \ref{sec:prelims} contains some basic preliminaries. In Section \ref{sec:rpdt}, we prove our main result, a lower bound on the RPDT complexity of a natural class of functions arising out of subspace designs. In Section \ref{sec:comm}, we state a few plausible conjectures and show that they imply a lower bound on the communication complexity of functions arising out of subspace designs composed with the XOR gadget. Finally, we end up with some open problems in Section \ref{sec:conclusion}.

\section{Preliminaries}\label{sec:prelims}

In this section, we provide some basic preliminaries needed for the paper. Section \ref{subsec:notation} starts off with some notation. Then in Section \ref{subsec:subspaces}, we present some basic facts about subspaces. Then we introduce the basics of our models of computations, parity decision trees and communication protocols in Section \ref{subsec:pdt_etc}. Finally, in Section \ref{subsec:fourier}, we present some basic concepts from Fourier analysis.

\subsection{Notation}\label{subsec:notation}

Given a subspace $S \subseteq \F_2^n$, we use $\dimension(S)$ to denote its dimension and $\codimension(S)$ to denote its codimension i.e. $n - \dimension(S)$. Given the standard bilinear form $\dotp{\cdot}{\cdot}$ on $\F_2^n$, we can define the dual space of $S$ as the set $\bra{\ell \in \F_2^n ~ \mid ~ \forall x \in S ~ \dotp{\ell}{x} = 0}$. It is a subspace of dimension $n-\dimension(S)$ and its dual space is $S$.

Given a subspace $S$ of dimension $k$, fix a basis $L = \bra{\ell_1,\dots,\ell_{n-k}}$ of its dual space. For every point $a \in \F_2^{n-k}$, we can define the set $S^{L}_{a} = \bra{x \in \F_2^n ~ \mid ~ \forall i \in [n-k] ~ \dotp{\ell_i}{x} = a_i}$. These are called affine shifts, or cosets, of $S$. Sets of the kind $S^{L}_{a}$ are also called affine subspaces. Each coset of $S$ also has size $2^k$. We can also define a coset map of $S$ with respect to a basis of its dual space as
   \[ \cosetmap_S^L(x) = (\dotp{\ell_1}{x}, \dots, \dotp{\ell_{n-k}}{x}). \]
It is easy to see that the choice of basis for the dual space does not affect the set of cosets of $S$. It merely affects the string $a \in \F_2^{n-k}$ that is used to refer to a specific coset. Hence we will refer to the coset map as $\cosetmap_S$, and we may choose an arbitrary basis of the dual space of $S$ in order to interpret the coset map.

From here on, we will use $\zone$ to refer to $\F_2$. The values $0$ and $1$ represent the additive and multiplicative identity of $\F_2$.

\subsection{Basic facts about subspaces}\label{subsec:subspaces}

Here we mention two facts about subspaces that will be useful. We include their proofs in \Cref{appendix:subspacefacts}.

\begin{restatable}[Disjoint Subspaces]{lemma}{subspacedisj}\label{lem:subspacedisj}
    Let $S$ be a subspace of $\zone^n$ of dimension $d_1$.
    Let $T$ be a subspace of $\zone^n$ of dimension $d_2$ chosen uniformly at random. Then $\Pr_{T}[S \cap T = \bra{0}] \geq 1 - n2^{d_1+d_2-n}$.
\end{restatable}

\begin{restatable}{lemma}{subspaceavoidance}\label{lem:subspaceavoidance}
    Let $V$ and $W$ be affine subspaces of $\zone^n$ satisfying \[ \frac{|V \cap W|}{|W|} < \frac{|V|}{2^n}. \] Then $V \cap W = \emptyset$.
\end{restatable}

\subsection{Parity decision trees, communication complexity and the corruption bound}\label{subsec:pdt_etc}

We now define parity decision trees, aimed at computing functions of the form $f: \zone^n \rightarrow \zone$.

\begin{definition}[Parity Decision Tree]
    A parity decision tree $T$ is a binary tree rooted at a node $r$ satisfying the following properties.
    \begin{itemize}
        \item Each internal node is labelled with a set $S \subseteq [n]$.
        \item Each internal node has two children, with one of the edges labelled with a $0$ and the other labelled with a $1$.
        \item Each leaf has a label from $\zone$.
    \end{itemize}
    
    A parity decision tree outputs a value $a \in \zone$ on given an input $x \in \zone^n$ as follows. The ``current node'' below is initialized to the root node $r$.
    \begin{itemize}
        \item The tree computes $b = \oplus_{i \in S} x_i$, where $S$ is the label on the current node.
        \item The tree moves to the child that is reached by taking the edge labelled $b$. If the child is a leaf, output the label of the leaf. Else, repeat the previous step with the child as the current node.
    \end{itemize}
    
    The cost of the parity decision tree is defined as the height of the tree.
\end{definition}

\begin{definition}[Randomized Parity Decision Tree]
    A randomized parity decision tree (RPDT) of cost $c$ is a distribution over deterministic parity decision trees of cost $c$. The output of the RPDT on an input $x$ is the random variable defined as the output of $T$ on $x$, where $T$ is a parity decision tree sampled as per the distribution specified by the RPDT.
\end{definition}

The $\epsilon$-error RPDT complexity of a function $f$, denoted $\RPDT{\epsilon}(f)$, is the minimum cost of an RPDT $T$ such that $\forall x, \Pr[f(x) = T(x)] \geq 1-\epsilon$.

\begin{lemma}[Corruption, RPDT version]\label{lem:corruptionrpdt}
    Let $f: \zone^n \rightarrow \zone$. Let $\mu$ be a distribution on $\zone^n$ such that $\mu(f^{-1}(0)) = \half$. Let $\epsilon \leq 1/8$.
    Then an $\epsilon$-error cost-$c$ RPDT computing $f$ implies the existence of an affine subspace $W$ such that
    \begin{itemize}
        \item $\mu(W \cap f^{-1}(1)) \leq 4\epsilon \mu(W)$ and
        \item $\codimension(W) \leq c$.
    \end{itemize}
\end{lemma}

\begin{proof}
    Note that an $\epsilon$-error cost-$c$ RPDT $T$ computing $f$ implies that for any distribution $\mu$ over the inputs of $f$, there is an RPDT whose expected error, $\E_{T,x \sim \mu}[|T(x) - f(x)|]$, is at most $\epsilon$. Since $T$ is a distribution over deterministic parity decision trees, there is a deterministic parity decision tree whose expected error is also at most $\epsilon$.

    Suppose that a subspace such as the one posited in the lemma statement did not exist. Then for any cost-$c$ parity decision tree $T$, we may compute the error made as follows. Note that the set of inputs that reach any specific leaf forms an affine subspace of codimension at most $c$, with each pair of such affine subspaces being disjoint. Let $\mathcal{L}$ be the set of these affine subspaces corresponding to the leaves of $T$ that are labelled $0$. Then $\sum_{V \in \mathcal{L}} \mu(V) \geq \half - \epsilon$, since otherwise $T$ would be outputting $1$ on more than an $\epsilon$ mass of $0$-inputs. But then $\sum_{V \in \mathcal{L}} \mu(V \cap f^{-1}(1)) \geq \sum_{V \in \mathcal{L}} 4\epsilon\mu(V) \geq 4\epsilon(\half - \epsilon) \geq 2\epsilon - 4\epsilon^2 > \epsilon$. So on more than an $\epsilon$ mass of $1$-inputs, $T$ outputs $0$. Hence the tree $T$ is erring on a larger than $\epsilon$ mass of inputs and we have a contradiction.
\end{proof}

We now move to communication complexity. We are concerned with the number of bits that two parties Alice and Bob need to communicate in order to compute a function $F: \mathcal{X} \times \mathcal{Y} \rightarrow \zone$. See \cite{KN97} for a thorough introduction to the topic. We will use that a deterministic communication protocol of cost $c$ partitions the input space of $F$ into at most $2^c$ rectangles (sets of the form $A \times B$ for $A \subseteq \mathcal{X}, B \subseteq \mathcal{Y}$), and it outputs the same value on all inputs in a rectangle. Randomized communication is defined akin to randomized parity decision trees.

\begin{definition}[Randomized Communication Protocol]
    A randomized communication protocol of cost $c$ is a distribution over deterministic communication protocols of cost $c$. The output of the randomized communication protocol on an input $x$ is the random variable defined as the output of $T$ on $(x,y)$, where $T$ is a communication protocol sampled as per the distribution specified by the randomized communication protocol.
\end{definition}

The $\epsilon$-error randomized communication complexity of a function $F$ is the minimum cost of an randomized communication protocol $T$ such that $\forall x,y, \Pr[F(x,y) = T(x,y)] \geq 1-\epsilon$.

The following is a lower-bound technique for randomized communication complexity akin to the lower bound for RPDTs given previously. This technique is well-known with roots in \cite{Yao83}.

\begin{lemma}[Corruption]\label{lem:corruption}
    Let $F: \zone^n \rightarrow \zone$. Let $\nu$ be a distribution on $\zone^n$ such that $\nu(F^{-1}(0)) = \half$. Let $\epsilon < 1/8$.
    Then an $\epsilon$-error cost-$c$ randomized communication protocol computing $F$ implies the existence of a rectangle $R$ such that
    \begin{itemize}
        \item $\nu(R \cap F^{-1}(1)) \leq 4\epsilon \nu(R)$ and
        \item $\nu(R) \geq 2^{-c-3}$.
    \end{itemize}
\end{lemma}



\subsection{Basic notions from Fourier analysis}\label{subsec:fourier}

We now move to Fourier analysis, a particularly useful tool in analyzing Boolean functions. We define the parity functions as follows. For each $S \subseteq [n]$, we define a parity function $\chi_S: \zone^n \rightarrow \pmone$ as $\chi_S(x) = (-1)^{\sum_{i \in S} x_i}$. These form an orthonormal basis for the class of functions from $\zone^n$ to $\R$ under the inner product $\langle f,g \rangle = \frac{1}{2^n} \sum_{x \in \zone^n} f(x)g(x)$. Hence every such function $f$ can be written as $\sum_S \hat{f}(S) \chi_S$. The values $\hat{f}(S)$ are referred to as Fourier coefficients and can be computed as $\langle f, \chi_S \rangle$. Let $\hat{f}$ denote the vector $(\hat{f}(S))_{S \subseteq [n]} \in \R^{2^n}$, known as the Fourier spectrum. We define the following measures of $f$.

\begin{itemize}
    \item The sparsity of $f$ is $\sparsity{f}$.
    \item The spectral norm of $f$ is $\spectral{f}$.
    \item The $\epsilon$-approximate sparsity of $f$, $\asparsity{f}{\epsilon}$, is $\min_{g : \forall x ~ |g(x) - f(x)| \leq \epsilon} \sparsity{g}$.
    \item The $\epsilon$-approximate spectral norm of $f$, $\aspectral{f}{\epsilon}$, is $\min_{g : \forall x ~ |g(x) - f(x)| \leq \epsilon} \spectral{g}$.
\end{itemize}

The Fourier spectrum of a subspace is easy to compute. (See, for instance, \cite{O14}.) It follows from the spectrum that any subspace $V \subseteq \zone^n$, the function $\mathds{1}_V$ satisfies $\spectral{\mathds{1}_V} = 1$.

For a function $f : \zone^n \rightarrow \R$ its composition with $\XOR$, denoted $f \circ \XOR$, is a function $F: \zone^n \times \zone^n \rightarrow \R$ defined as $F(x,y) = f(x \oplus y)$ where $x \oplus y$ is the bitwise $\XOR$ of $x$ and $y$.

It is a well known fact that for a function $F := f \circ \XOR$, the rank of the communication matrix of $F$, denoted $\mathsf{rank}(F)$, is equal to $\sparsity{f}$. The $\epsilon$-approximate rank of $F$ is at most the $\epsilon$-approximate sparsity of $f$.

We note a theorem useful in showing that a function has small approximate sparsity.

\begin{theorem}[Grolmusz's Theorem~\cite{BS90,Gro97,Zhang14,CMS20}]\label{thm:grolmusz}
    For any $f : \zone^n \rightarrow \zone$ and $\delta > \epsilon \geq 0$,
    \[ \asparsity{f}{\delta} \leq O\left(\aspectral{f}{\epsilon}^2 n/(\delta-\epsilon)^2\right). \]
\end{theorem}

We conclude the preliminaries with the useful notion of entropy.

\begin{definition}[Entropy]
  Let $X$ be a discrete random variable. The entropy $H(X)$ is defined as
  \[
    H(X) \defeq \sum_{s \in \textnormal{supp(X)}} \Pr[X=s] \log \left(\frac{1}{\Pr[X=s]}\right).
  \]
\end{definition}
    
\begin{fact}[Folklore]
    $|\textnormal{supp}(X)| = k \implies H(X) \leq \log k$, with equality if and only if $X$ is uniform.
\end{fact}
    

\section{The RPDT Complexity of Dual Subspace Designs}\label{sec:rpdt}

In this section, we prove a lower bound on the RPDT complexity of a natural class of functions arising from subspace designs. A subspace design is a set of subspaces such that any small dimensional subspace non-trivially intersects only a few members of the set. (These are referred to as weak subspace designs in~\cite{GK16}.)

\begin{definition}[Subspace Design]
    An $n$-dimensional $(s,\A)$-subspace design is a set of subspaces $\bra{S_1, S_2, \cdots, S_m}$ of $\zone^n$ such that for all subspaces $T$ of dimension at most $s$, at most $\A$ of the $m$ subspaces intersect $T$ non-trivially.
\end{definition}

We call a set of subspaces $\bra{V_1, V_2, \cdots, V_m}$ of $\zone^n$ an $n$-dimensional $(s,\A)$-dual subspace design if their duals form an $(s,\A)$-subspace design. Dual subspace designs have an alternate characterization based on the notion of independent subspaces.

\begin{definition}[Independent Subspaces]
    Subspaces $S,T \subseteq \zone^n$ are independent if their coset maps are independent. That is, let $L_S$ and $L_T$ be arbitrary bases for the dual spaces of $S$ and $T$. For a variable $x$ chosen uniformly at random from $\zone^n$, consider the random variables $\cosetmap_S(x)$ and $\cosetmap_T(x)$. For every $a \in \F_2^{\codimension(S)},b \in \F_2^{\codimension(T)}$, we want that $\Pr[\cosetmap_S(x)=a | \cosetmap_T(x)=b] = \Pr[\cosetmap_S(x)=a] = 2^{-\codimension(S)}$.

    In particular this implies that every coset of $S$ intersects with every coset of $T$.
\end{definition}

We now state the alternate characterization of dual subspace designs.

\begin{claim}\label{clm:subspace_design_ind}
    The set $\bra{V_1, V_2, \cdots, V_m}$ of $\zone^n$ is an $n$-dimensional $(s,\A)$-dual subspace design if and only if for all subspaces $W$ of \emph{codimension} at most $s$, at least $m-\A$ of the $m$ subspaces are \emph{independent} from $W$.
\end{claim}

This claim follows from the following lemma relating trivial subspace intersections and independent subspaces.

\begin{lemma}[Independent Subspaces]
    Subspaces $S$ and $T$ of $\F_2^n$ are independent if and only if the dual space of $S$ and the dual space of $T$ intersect trivially (i.e. only at the point $0 \in \F_2^n$).
\end{lemma}

\begin{proof}
    Let $V$ and $W$ be the dual spaces of $S$ and $T$ respectively. If $V$ and $W$ intersected at a non-zero point $\ell \in \F_2^n$, then consider bases $L_S$ and $L_T$ for $V$ and $W$ respectively, wherein $\ell$ is the first element of $L_S$ and also the first element of $L_T$. The coset maps of $S$ and $T$ with this choice of $L_S$ and $L_T$ cannot be independent since for all $x \in \F_2^n$, the first entries of $\cosetmap_S^{L_S}(x)$ and $\cosetmap_T^{L_T}(x)$ will always agree.
    
    For the other direction, let $L_S$ and $L_T$ be arbitrary bases for $V$ and $W$ respectively. We will show that if $V$ and $W$ intersect trivially, then the coset maps are independent. Assuming $V$ and $W$ intersect trivially, this means that $\spa(L_S) \cap \spa(L_T) = \bra{0}$. Hence $L = L_S \cup L_T$ is an independent set of size $\dimension(V) + \dimension(W)$. Consider the subspace $X$ with basis $L$, and let $R$ be its dual subspace. The cosets of $R$ each have size $2^{n - \dimension(V) - \dimension(W)}$. For any $a \in \F_2^{\codimension(S)},b \in \F_2^{\codimension(T)}$, the set $\bra{x ~ \mid ~ \cosetmap_S^{L_S}(x) = a \wedge \cosetmap_T^{L_T}(x) = b}$ is a coset of $R$. Hence $\Pr[\cosetmap_S^{L_S}(x)=a | \cosetmap_T^{L_T}(x)=b] = 2^{- \dimension(V) - \dimension(W)}/2^{-\dimension(W)} = 2^{-\dimension(V)}$
\end{proof}

A useful corollary of Claim \ref{clm:subspace_design_ind} is that an $(s,\A)$-dual subspace design also forms a hitting set for the set of all affine subspaces of codimension at most $s$. We will use this fact to lower bound the randomized parity decision tree complexity of unions of subspaces.

\begin{corollary}\label{cor:dualdesignsashittingsets}
        Let $\bra{V_1, V_2, \cdots, V_m}$ be an $n$-dimensional $(s,\A)$-dual subspace design. For all affine subspaces $W$ of \emph{codimension} at most $s$, at least $m-\A$ of the $m$ subspaces intersect with $W$.
\end{corollary}

\begin{proof}
    This follows from Claim \ref{clm:subspace_design_ind} and the fact that if two subspaces $S$ and $T$ are independent, then $S$ will intersect any affine shift of $T$ non-trivially.
\end{proof}

We are now ready to prove the main theorem of the section.

\begin{theorem}\label{thm:dualsubspacehardness}
    Let $\mV$ be an $n$-dimensional $(s,\A)$-dual subspace design of size $m$.
    
    Let $f$ be the function defined as $f^{-1}(1) = \bigcup_{V \in \mV} V$. We now show that $\RPDT{\epsilon}(f) \geq s$ as long as $\epsilon < \frac{m-\A}{8m} \frac{|f^{-1}(0)|}{2^{n}}$.
\end{theorem}

\begin{proof}
    Consider the distribution $\mu$ defined over the inputs of $f$ as follows.
    \begin{itemize}
        \item Sample $z \sim_{\mathsf{unif}} \zone$.
        \item If $z=0$, output a uniformly random input from $f^{-1}(0)$.
        \item Otherwise, sample $V \sim_{\mathsf{unif}} \mV$.
        \item Output a uniformly random input from $V$.
    \end{itemize}

    Assuming that $f$ is computed by an $\epsilon$-error cost $c$ RPDT, \Cref{lem:corruptionrpdt} implies the existence of a subspace $W$ such that
    \begin{itemize}
        \item $\mu(W \cap f^{-1}(1)) \leq 4\epsilon \mu(W)$ and
        \item $\codimension(W) \leq c$.
    \end{itemize}

    Assume we have a $W$ such that $\mu(W \cap f^{-1}(1)) \leq 4\epsilon \mu(W)$. This means that $\mu(W \cap f^{-1}(1)) \leq \frac{4\epsilon}{1-4\epsilon} \mu(W \cap f^{-1}(0))$. We also know the following from the definition of $\mu$.
    \begin{align*}
        \mu(W \cap f^{-1}(1)) &= \half \cdot \frac{1}{|\mV|} \sum_{V \in \mV} \frac{|W \cap V|}{|V|}\\
        \mu(W \cap f^{-1}(0)) &= \half \cdot \frac{|W \cap f^{-1}(0)|}{|f^{-1}(0)|} \leq \half \cdot \frac{|W|}{|f^{-1}(0)|}
    \end{align*}
    Putting these together, we get that \[ \frac{1}{|\mV|} \sum_{V \in \mV} \frac{|W \cap V|}{|V|} \leq \frac{4\epsilon}{1-4\epsilon} \frac{|W|}{|f^{-1}(0)|}. \]

    Now if $\epsilon < \frac{m-\A}{8m} \frac{|f^{-1}(0)|}{2^{n}} \leq \frac{1}{8}$, then $\frac{4\epsilon}{1-4\epsilon} < \frac{m-\A}{m} \frac{|f^{-1}(0)|}{2^{n}}$.
    This implies that less than $m-\A$ subspaces of $\mV$ can satisfy $\frac{|W \cap V|}{|V|} \geq \frac{|W|}{2^{n}}$, and hence more than $\A$ of them \emph{must} satisfy $\frac{|W \cap V|}{|V|} < \frac{|W|}{2^{n}}$. This means that $W \cap V = \emptyset$ (\Cref{lem:subspaceavoidance}).
    In other words, $W$ is an affine subspace that managed to evade more than $\A$ subspaces of $\mV$. But by \Cref{cor:dualdesignsashittingsets}, if $W$ is of codimension at most $s$, then it is disjoint from at most $\A$ subspaces of $\mV$. So $W$ must be of codimension more than $s$.

    Hence the codimension of $W$, and thereby the cost of the RPDT, is at least $s$.
\end{proof}

\begin{remark}
    The above proof would also work for any union of affine subspaces which forms a hitting set for the set of all large affine subspaces the way that the dual subspace design does.
\end{remark}


\subsection{Narrowing the gap between RPDT complexity and approximate sparsity to cubic}

In this section, we instantiate Theorem \ref{thm:dualsubspacehardness} with random subspaces to get a mere cubic gap between RPDT complexity and approximate sparsity. It is known that there are efficient probabilistic constructions of subspace designs. We go through such a construction here, and use it to show our main theorem.

\begin{theorem}\label{thm:randomspaces}\label{cor:separation}
    Let $m = 100n$. Let $V_1, V_2, \dots, V_m$ be subspaces of $\zone^n$ chosen independently and uniformly at random from the set of subspaces of dimension $2n/5$. With probability $1-o(1)$ the following two statements are true.
    \begin{itemize}
        \item $\mV = \bra{V_1,\dots,V_m}$ forms an $(n/5,m/10)$-dual subspace design.
        \item Every pair of subspaces in $\mV$ intersects trivially.
    \end{itemize}
\end{theorem}

\begin{proof}
    Let $W$ be a fixed affine subspace of $\zone^n$ of dimension $4n/5$. Let $\mV = \bra{V_1, V_2, \cdots, V_m}$ be subspaces of $\zone^n$ chosen independently and uniformly at random from the set of subspaces of dimension $2n/5$.

    Since the duals of $W$ and $V_1$ have dimension $3n/5$ and $n/5$ respectively, the probability that $W$ and $V_1$ are independent is at least $1 - n{2^{-n/5}}$ (\Cref{lem:subspacedisj}). This is independently true of $W$ and each $V \in \mV$. The probability that $W$ is not independent with \emph{at least} $m/10$ of the $m$ subspaces is at most $\binom{m}{m/10} (n 2^{-n/5})^{m/10}$.

    Since the number of subspaces of dimension $4n/5$ is at most $(2^n)^{4n/5} = 2^{4n^2/5}$, the probability that there exists such a subspace $W$ that is not independent with at least $m/10$ of the subspaces in $\mV$ is at most $2^{4n^2/5} \binom{m}{m/10} (n 2^{-n/5})^{m/10}$.
    
    Setting $m = 100n$, this upper bound is at most $2^{.8n^2 + 100n + 10n \log n - 2n^2} = o(1)$.

    Hence with high probability, $\mV$ is an $(n/5,m/10)$-dual subspace design.
    
    Let $f$ be defined as in the theorem statement. Note that since $V_1$ and $V_2$ are random subspaces of dimension $2n/5$, the probability that they intersect only at $0$ is at least $1-n2^{-n/5}$. The probability that any two subspaces in $\mV$ intersect at more than just $0$ is at most $\binom{m}{2} n 2^{-n/5} = o(1)$.
\end{proof}

\separation*

\begin{proof}
    We know from \Cref{thm:randomspaces} that with probability $1-o(1)$ the set $\mV$ forms an $(n/5,m/10)$-dual subspace design. We also can trivially lower bound $|f^{-1}(0)|/2^n$ by $1 - m2^{-3n/5}$. Since $\mV$ is an $(n/5,m/10)$-dual subspace design, we can conclude from \Cref{thm:dualsubspacehardness} that for $\epsilon \leq 1/10$, $\RPDT{\epsilon}(f) \geq n/5$. 

    We also know from \Cref{thm:randomspaces} that with probability $1-o(1)$, every pair of subspaces from $\mV$ intersects trivially. When this event holds, $f$ can be represented as $\sum_{V \in \mV} \mathds{1}_{V} - (m-1) \mathds{1}_{V_0}$ where $V_0 = \bra{0}$ is the trivial subspace of dimension $0$. Since the spectral norm of a subspace is equal to $1$, the spectral norm of $f$ is upper bounded by $m + m-1 < 2m$. Using \Cref{thm:grolmusz}, this also implies that $\asparsity{f}{\epsilon} \leq O(m^2 n/\epsilon^2) = O(n^3)$ for any constant $\epsilon$.
    
    This concludes the proof of the merely cubic gap.
\end{proof}

\subsection{On Extending this to Communication}\label{subsec:communication}\label{sec:comm}


In this section, we state a plausible conjecture that would imply a lower bound on the randomized communication complexity of XOR compositions of our functions. The proof of this implication is in \Cref{appendix:commlowerbound}.

In the RPDT lower bound, we showed that in order for an affine subspace to avoid most of the subspaces of a dual subspace design, the codimension of the affine subspace needs to be large. We could hope for a similar statement in the communication world: For a rectangle to put very little mass on most of the subspaces making up a dual subspace design (i.e., puts very little mass on inputs $(x,y)$ such that $x \oplus y$ lies in the subspaces), the mass of the rectangle must be $2^{-\Omega(n)}$. One particularly neat conjecture that would imply that statement is the following, in which $\mU_k$ denotes the uniform distribution over $k$ elements.

\begin{conjecture}\label{conj:entropyloss}
    There exist constants $0 < \alpha < 1$, $\beta > 0$ and $k \geq 1$ such that the following holds. Let $\mV = \bra{V_1, \dots, V_m}$ be an $n$-dimensional $(s,\A)$-dual subspace design. Let $B_i$ be the coset map of $V_i$. Let $X$ be a random variable over $\zone^n$ such that $\| B_i(X) - \mU_{2^{\codimension(V_i)}} \|_1 \geq \alpha$ for more than $k\A$ values of $i \in [m]$. Then $H(X) \leq n - \beta s$.
\end{conjecture}

The merely cubic gap in the RPDT world used random subspaces. So for extending it to communication, it would be okay for us to bypass dual subspace designs and prove the theorem for random subspaces instead.

\begin{conjecture}\label{conj:entropyloss_randomsubspaces}
    There exists a constant $0 < \alpha < 1, \beta > 0$ such that the following holds. Let $m=100n$. Let $V_1, V_2, \dots, V_m$ be random subspaces of $\zone^n$ of dimension $2n/5$, and let $B_1, B_2, \cdots, B_m$ be their coset maps. Let $X$ be a random variable over $\zone^n$ such that $\| B_i(X) - \mU_{2^{3n/5}} \|_1 \geq \alpha$ for at least $m/3$ values of $i \in [m]$. Then with high probability, $H(X) \leq n - \beta n$.
\end{conjecture}

First of all note that the conjectures are true when $X$ is the uniform distribution over an affine subspace. To see this, suppose $X$ is the uniform distribution over an affine subspace $W$. $H(X) \ge n - s$ is the same as saying that $\codimension(W) \le s$. Then by Claim \ref{clm:subspace_design_ind}, for at least $m-h$ of the subspaces $V_1,\ldots, V_m$, $V_i$ and the dual space of $W$ are independent, which implies that $B_i(X)$ will be exactly uniform ($\mU_{2^{\codimension(V_i)}}$).

We discuss now why the Conjectures \ref{conj:entropyloss} and \ref{conj:entropyloss_randomsubspaces} appear to be a bit tricky to prove. While the conjectures are true for affine subspaces, the number of distributions (or even the number of subsets of $\zone^n$) are much larger (doubly exponential in $n$), so the conjectures are a leap of faith in this sense. But we haven't been able to come up with counterexamples and it would be very interesting to do so. The conceptual way to view the conjectures, e.g. Conjecture \ref{conj:entropyloss_randomsubspaces} to be concrete, is that if a random variable $X$ has the property that when projected down to $2n/5$ bits in various ways it loses $\Omega(1)$ bits of entropy, then $X$ overall loses $\Omega(n)$ bits of entropy. Shearer's lemma talks about these kind of statements. While in Shearer's lemma, the projections are onto subcubes, there are generalizations called Brascamp-Lieb inequalities which talk about more general projections (e.g. see \cite{christ2013optimal}). However, the Brascamp-Lieb inequalities can at best guarantee an $\Omega(n/k)$-bit entropy loss in $X$ if there is an $\Omega(1)$-bit entropy loss while projecting $X$ to $k$ bits in various ways. What we want is much stronger. This is one difficulty.

The other difficulty is that a Fourier type approach doesn't seem to work either. One can control $\| B_i(X) - \mU_{2^{\codimension(V_i)}} \|_1$ by bounding the $\ell_2$ distance and then trying to bound the Fourier coefficients of the distribution of $X$ on the dual space of $V_i$. But this doesn't give any meaningful bound (if done in a naive way at least).

We now state the lower bound on the randomized communication complexity of a dual subspace design composed with $\XOR$ that we get assuming \Cref{conj:entropyloss}. For a set of subspaces in $n$ dimensions $\mV = \bra{V_1, V_2, \dots, V_m}$, let $f_{\mV}$ be the function on $n$ bits that outputs $1$ on inputs in $\cup_{V \in \mV} V$.

\begin{theorem}\label{thm:commlowerbound}[Proof in \Cref{appendix:commlowerbound}]
    Let us assume \Cref{conj:entropyloss} holds with constants $\alpha, \beta$ and $k$. Let $\mV = \bra{V_1, V_2, \dots, V_m}$ be an $n$-dimensional $(s,\A)$-dual subspace design and define $\gamma$ so that $|\cup_{V \in \mV} V| = \gamma 2^n$.  Let $F = f_{\mV} \circ \XOR$. For $\epsilon < \frac{(1-\alpha)^2}{4} \frac{m-2k\A}{8m}(1-\gamma)$, the $\epsilon$-error randomized communication complexity of $F$ is at least $\beta s + \log(1-\gamma)$.
\end{theorem}

Given this lower bound, we would want to apply it to get a merely cubic gap between randomized communication complexity and approximate rank along the lines of \Cref{cor:separation}.

\begin{corollary}
    Let $\mV = \bra{V_1, V_2, \dots, V_m}$ be an $(n/5,m/20k)$-dual subspace design with (1) $m = 200kn$, (2) each subspace having dimension $2n/5$ and (3) every pair of subspaces intersecting trivially. Let $F = f_{\mV} \circ \XOR$. Then assuming \Cref{conj:entropyloss},
    \begin{itemize}
        \item The $1/10$-error randomized communication complexity of $F$ is $\Omega(n)$.
        \item $\mathsf{rank}_{1/10}(F) = O(n^3)$.
    \end{itemize}
\end{corollary}

\begin{proof}
    The size of $F^{-1}(1)$ would be at most $2^n \sum_{V \in \mV} |V| \leq 2^{n+2n/5} m = o(2^{2n})$. We can then use \Cref{thm:commlowerbound} to get a lower bound of $\beta n/5$ when $\epsilon < \frac{(1-\alpha)^2}{4} \frac{m-2k\A}{8m} \frac{|F^{-1}(0)|}{2^{2n}}$, which is a constant. Since we can use error reduction to go from error $1/10$ to any small constant error with only a constant blow-up in cost, the $1/10$-error randomized communication complexity is also $\Omega(n)$.
    
    The $\epsilon$-approximate rank of $f \circ \XOR$ is known to be at most the $\epsilon$-approximate sparsity of $f$. As analyzed in \Cref{cor:separation}, $\spectral{f_{\mV}} \leq 2m$ and $\asparsity{f_{\mV}}{1/10} \leq O(m^2n) = O(n^3)$ and hence $\mathsf{rank}_{1/10}(F) \leq O(n^3)$.
\end{proof}

The existence of a dual subspace design as required in the previous corollary follows by changing \Cref{thm:randomspaces} to set $m = 200kn$. The proof of the modified statement is syntactically identical to the proof of the original statement.

\section{Conclusion and open problems}\label{sec:conclusion}

We come up with new and improved refutations of the query complexity analogue of the log-approximate-rank conjecture, following the work of Chattopadhyay, Mande and Sherif \cite{CMS20}. Our examples are derived from subspace designs, a concept which has previously found applications in coding theory and pseudorandomness \cite{GX12, GK16, GXY17}. A lot of interesting open problems arise from our work, some of which we mention below.

\begin{enumerate}
    \item \textbf{(Communication complexity of XOR composed subspace designs).} What is the randomized communication complexity of dual subspace designs composed with XOR (as studied in Section \ref{subsec:communication})? A lower bound would follow from Conjecture \ref{conj:entropyloss}. If Conjecture \ref{conj:entropyloss} is false, is there an alternate way to prove the communication lower bound? Since we already have an RPDT lower bound for dual subspace designs, these functions provide a interesting class of functions to study randomized XOR lifting. Currently we cannot even prove that this class of functions do not have large monochromatic rectangles.
    \item \textbf{(Communication complexity of XOR composed random subspaces).} What is the randomized communication complexity of random subspaces composed with XOR? A lower bound would follow from Conjecture \ref{conj:entropyloss_randomsubspaces} which follows from Conjecture \ref{conj:entropyloss}. Even if Conjecture \ref{conj:entropyloss} is false, Conjecture \ref{conj:entropyloss_randomsubspaces} could still be true or perhaps easier to prove. If even Conjecture \ref{conj:entropyloss_randomsubspaces} is false, is there an alternate way to prove the communication lower bound, perhaps adapting the technique of \cite{HHL16} to the randomized communication setting? Here also we cannot prove that there are no large monochromatic rectangles.
    \item \textbf{(Quantum communication complexity of XOR composed subspace designs).} What is the quantum communication complexity of dual subspace designs composed with XOR? Is there a function in this class which has polylogarithmic quantum communication complexity? 
    \item \textbf{(RPDT and approximate sparsity).} What is the optimal gap between RPDT complexity and approximate sparsity? We give examples where the RPDT complexity is at least cube root of the approximate sparsity and also RPDT complexity is easily seen to be at most the approximate sparsity.
\end{enumerate}

\bibliographystyle{plainurl}
\bibliography{bibliography.bib}

\appendix
\section{Facts About Subspaces}\label{appendix:subspacefacts}

\subspacedisj*

\begin{proof}
    Let us generate $T$ by choosing $d_2$ vectors $\bra{v_1,\dots,v_{d_2}}$, each vector independent of the previous ones, in order to form a basis for $T$. The subspace $S$ intersects $T$ trivially if and only if for all $i \in [d_2]$, $v_i \not\in \spa(\bra{v_j}_{j < i} \cup S)$. We call these events $E_1, \dots, E_{d_2}$. When choosing $v_i$ to add to the basis for $T$, there are $2^n - 2^{i-1}$ choices, since $|\spa(\bra{v_j}_{j < i})| = 2^{i-1}$. Conditioned on $E_1, \dots, E_{i-1}$, we also know that $|\spa(\bra{v_j}_{j < i} \cup S)| = 2^{i-1+d_1}$. The probability of $E_i$ occurring is \[ \frac{|\left(\zone^n \setminus \spa(\bra{v_j}_{j < i}) \right) \setminus \spa(\bra{v_j}_{j < i} \cup S)|}{|\zone^n \setminus \spa(\bra{v_j}_{j < i})|} = \frac{|\zone^n \setminus \spa(\bra{v_j}_{j < i} \cup S)|}{|\zone^n \setminus \spa(\bra{v_j}_{j < i})|}.\] We can then calculate the probability of $S \cap T = \emptyset$ as

    \begin{align*}
    \Pr\left[\bigcap_{i \in [d_2]} E_i\right] &= \prod_{i = 1}^{d_2} \Pr\left[ E_i ~ \mid ~ E_1, \cdots, E_{i-1} \right] = \prod_{i = 1}^{d_2} \frac{2^n - 2^{d_1+i-1}}{2^n-2^{i-1}}\\
    &\geq \left(1 - \frac{2^{d_1+d_2}}{2^n}\right)^{d_2} \geq 1 - \frac{d_2}{2^{n-d_1-d_2}}.
    \end{align*}
\end{proof}

\subspaceavoidance*

\begin{proof}
    Let $\bra{\dotp{v_i}{x} = a_i}_{i \in [k]}$ be the constraints defining the affine subspace $W$. Let $W_0, W_1, \cdots, W_k$ be the affine spaces defined as follows. The constraints for $W_j$ are $\bra{\dotp{v_i}{x} = a_i}_{i \in [j]}$. Clearly $W_0 = \zone^n$ and $W_k = W$.
    
    Now let us assume that $|V \cap W_i| \neq 0$ and is hence an affine subspace. The set $V \cap W_{i+1}$ is the same affine subspace with the added constraint $\dotp{v_{i+1}}{x} = a_{i+1}$.
    \begin{itemize}
        \item If this constraint was already implied by the constraints in $V \cap W_i$, then $|V \cap W_{i+1}| = |V \cap W_i|$.
        \item If this constraint is incompatible with the constraints in $V \cap W_i$, then $|V \cap W_{i+1}| = 0$.
        \item If this constraint was independent of the constraints in $V \cap W_i$, then $|V \cap W_{i+1}| = |V \cap W_i|/2$.
    \end{itemize}
    
    Hence $|V \cap W_k|$ is either $0$ or is at least $|V \cap W_0|/2^k$. On the other hand, $|W|/2^n = 1/2^k$. Since $V \cap W_k = V \cap W$ and $V \cap W_0 = V$, we can rewrite this as
    \[ V \cap W \neq \emptyset \implies \frac{|V \cap W|}{|V|} \geq \frac{|W|}{2^n}. \] 
\end{proof}

\section{Randomized Communication Lower Bound Assuming the Conjecture}\label{appendix:commlowerbound}

In the following lower bound, we assume \Cref{conj:entropyloss} to hold with $\alpha=\half$. After the proof we discuss how to modify it to hold for other values of $\alpha$.

\begin{theorem}
    Let us assume \Cref{conj:entropyloss} holds with $\alpha = \half$ and some constants $\beta, k$. Let $\mV = \bra{V_1, V_2, \dots, V_m}$ be an $n$-dimensional $(s,\A)$-dual subspace design and define $\gamma$ so that $|\cup_{V \in \mV} V| = \gamma 2^n$.  Let $F = f_{\mV} \circ \XOR$. For $\epsilon < \frac{m-2k\A}{128m}(1-\gamma)$, the $\epsilon$-error randomized communication complexity of $F$ is at least $\beta s + \log(1-\gamma)$.
\end{theorem}

\begin{proof}
    For any $V \in \mV$, let $S_V = \bra{(x,y) \in \zone^n \times \zone^n \mid x \oplus y \in V}$. Note that $|S_V| = 2^n |V|$ and $F^{-1}(1) = \cup_{V \in \mV} S_V$.
    Consider the distribution $\nu$ defined over the inputs of $F$ as follows.
    \begin{itemize}
        \item Sample $z \sim_{\mathsf{unif}} \zone$.
        \item If $z=0$, output a uniformly random input from $F^{-1}(0)$.
        \item Otherwise, sample $V \sim_{\mathsf{unif}} \mV$.
        \item Output a uniformly random input from $S_V$.
    \end{itemize}

    Assuming $F$ is computed by an $\epsilon$-error cost $c$ communication protocol, \Cref{lem:corruption} implies the existence of a rectangle $R$ such that
    \begin{itemize}
        \item $\nu(R \cap F^{-1}(1)) \leq 4\epsilon \nu(R)$ and
        \item $\nu(R) \geq 2^{-c-3}$.
    \end{itemize}

    Assume we have an $R$ such that $\nu(R \cap F^{-1}(1)) \leq 4\epsilon \nu(R)$. This means that $\nu(R \cap F^{-1}(1)) \leq \frac{4\epsilon}{1-4\epsilon} \nu(R \cap F^{-1}(0))$. We also know the following from the definition of $\nu$.
    \begin{align*}
        \nu(R \cap F^{-1}(1)) &= \half \cdot \frac{1}{|\mV|} \sum_{V \in \mV} \frac{|R \cap S_V|}{|S_V|}\\
        \nu(R \cap F^{-1}(0)) &= \half \cdot \frac{|R \cap F^{-1}(0)|}{|F^{-1}(0)|} \leq \half \cdot \frac{|R|}{|F^{-1}(0)|}
    \end{align*}
    Putting these together, we get that \[ \frac{1}{|\mV|} \sum_{V \in \mV} \frac{|R \cap S_V|}{|S_V|} \leq \frac{4\epsilon}{1-4\epsilon} \frac{|R|}{|F^{-1}(0)|}. \]

    Now if $\epsilon < \frac{m-2k\A}{128m} \frac{|F^{-1}(0)|}{2^{2n}} < 1/8$, then $\frac{4\epsilon}{1-4\epsilon} < \frac{m-2k\A}{16m} \frac{|F^{-1}(0)|}{2^{2n}}$.
    This implies that less than $m-2k\A$ subspaces of $\mV$ can satisfy $\frac{|R \cap S_V|}{|S_V|} \geq \frac{|R|}{16 \cdot 2^{2n}}$, and hence more than $2k\A$ of them \emph{must} satisfy $\frac{|R \cap S_V|}{|S_V|} < \frac{|R|}{16 \cdot 2^{2n}}$. Let us fix such a $V$.

    Let $\cosetmap_V$ denote the function $\cosetmap_V^{L_V}$ for some fixed basis $L_V$ of the dual space of $V$. Let $R = A \times B$. Then $\frac{|R \cap S_V|}{|R|}$ is the probability that, when $x$ and $y$ are sampled uniformly at random from $A$ and $B$, $\cosetmap_V(x) = \cosetmap_V(y)$. Let $A_V$ be the distribution of $\cosetmap_V(x)$ and $B_V$ be the distribution of $\cosetmap_V(y)$. The condition $\frac{|R \cap S_V|}{|R|} < \frac{|S_V|}{16 \cdot 2^{2n}}$ can be rewritten as \[ \Pr_{x' \sim A_V, y' \sim B_V}[x' = y'] < \frac{|S_V|}{16\cdot 2^{2n}} = \frac{1}{16 \cdot 2^{\codimension(V)}}. \]
     It follows that $A_V(S) < 1/4$ where $S = \bra{y' \mid B_V(y') \geq \frac{1}{4\cdot 2^{\codimension{V}}}}$. However, $B_V(S)$ must be at least $\sfrac{3}{4}$, since $B_V(\overline{S}) \leq \sfrac{1}{4}$.

    Hence $A_V$ and $B_V$ have total variational distance at least $\half$, and $\| A_V - B_V \|_1 \geq 1$. By the triangle inequality, $\max\{ \| A_V - \mU_{2^{\codimension(V)}} \|_1, \| B_V - \mU_{2^{\codimension(V)}} \|_1 \} \geq \half$.

    Hence, either there are more than $k\A$ subspaces that satisfy $\| A_V - \mU_{2^{\codimension(V)}} \| \geq \half$ or there are more than $k\A$ subspaces that satisfy $\| B_V - \mU \| \geq \half$. Without loss of generality we assume the former. Now we use our conjecture. The conjecture implies that $H(A) \leq n- \beta s$. Hence $\frac{|R|}{2^{2n}} \leq 2^{-\beta s}$.

    We now want to move from $|R|$ being small under the uniform distribution to $R$ being small under $\nu$. We know that $\nu(R \cap F^{-1}(1)) \leq 4\epsilon\nu(R) < \nu(R)/2$, so $\nu(R \cap F^{-1}(0)) \geq \nu(R)/2$. We also know from the definition of $\nu$ that \[ \nu(R \cap F^{-1}(0)) = \frac{|R \cap F^{-1}(0)|}{2|F^{-1}(0)|} \leq \frac{|R|}{2 \cdot 2^{2n}} \cdot \frac{2^{2n}}{|F^{-1}(0)|} \leq 2^{-\beta s-1} \cdot \frac{1}{1-\gamma}. \] So $\nu(R) \leq 2\nu(R \cap F^{-1}(0)) \leq 2^{-\beta s - 1 - \log(1-\gamma)} $. Hence the cost of the protocol is at least $\beta s + \log(1-\gamma) - 3$.
\end{proof}

We now explain how to modify the proof assuming the conjecture were true for other values of $\alpha$. Then the theorem statement would be modified, setting $\epsilon < \frac{(1-\alpha)^2}{4} \frac{m-2k\A}{8m} (1-\gamma)$. The proof would go through as it does above, analyzing a rectangle $R = A \times B$.

\begin{itemize}
    \item We would find more than $2k\A$ subspaces $V$ such that $\Pr[A_V = B_V] < \frac{(1-\alpha)^2}{4} \frac{|S_V|}{2^{2n}}$ as is done in the above proof.
    \item We would then set $S = \bra{y' \mid B_V(y') \geq \frac{1-\alpha}{2 \cdot 2^{\codimension(V)}}}$. This would mean that $A_V(S) \leq \frac{1-\alpha}{2}$ and $B_V(S) \geq 1 - \frac{1-\alpha}{2}$. Hence $\| A_V - B_V \|_1 \geq 2\alpha$, and one of $A$ or $B$ (wlog, $A$) satisfies $\| A_V - \mU_{2^{\codimension(V)}} \|_1 \geq \alpha$ for at least $k\A$ subspaces from the dual subspace design.
    \item The proof would continue as it does above, using the conjecture to conclude that the cost of the protocol would be at least $\beta s + \log(1-\gamma) - 3$, which is $\Omega(s)$ for constant $\gamma$.
\end{itemize}

\end{document}